\documentclass[11pt]{amsart}

\usepackage{latexsym,amsfonts,amsmath}
   \usepackage{amsfonts}
\usepackage{graphicx,graphics}
\usepackage{amsmath, amssymb, amsthm, amscd, amsfonts}
\usepackage{graphicx}%
\usepackage{fancyhdr, geometry}
\usepackage{hyperref,url}
\usepackage{color,verbatim}
\usepackage[normalem]{ulem}
\usepackage{colordvi}
\usepackage[ruled,  lined,  commentsnumbered,norelsize]{algorithm2e}
\newtheorem{thm}{Theorem}[section] 
\newtheorem{lemma}[thm]{Lemma}

\newtheorem{corollary}[thm]{Corollary}

\newtheorem{lm}[thm]{Lemma}

\theoremstyle{definition}
\newtheorem{defn}[thm]{Definition}

\newtheorem{obs}[thm]{Observation}
\newtheorem{remark}[thm]{Remark}
\newtheorem{rmk}[thm]{Remark}

\newtheorem{example}[thm]{Example}

\newcommand{\R}{{\mathbb{R}}}

\newcommand{\RR}{\mathfrak{R}}

\renewcommand{\P}{\mathcal{P}}
\newcommand{\PP}{\mathbb P}

\newcommand{\C} {\mathbb C}

\newcommand{\mH}{H}

\newcommand{\mB}{\mathcal{B}}

\newcommand{\ideal}[1]{( {#1})}  
\newcommand{\vecsp}[1]{\langle {#1}\rangle} 
\newcommand{\vi}{{\sf V}}  

\newcommand{\Basis}{\mathrm{Basis}}
\newcommand{\BruteForce}{\textrm{BruteForce}}

\DeclareMathOperator{\V}{\mathcal V}
\DeclareMathOperator{\coeff}{coeff}

\def\jdlqed{\vbox{\hrule \hbox{\vrule\hbox to
5pt{\vbox to 6pt{\vfil}\hfil}\vrule}\hrule}}

\newcommand{\despina}[1]{\marginpar{\it\textcolor{red}{[D: #1]}}}
\newcommand{\sonja}[1]{\marginpar{\it\textcolor{red}{ [S: #1]}}}

\begin{document}

\title[Helly Numbers and Violator Spaces in Computational  Algebra]{Random Sampling in Computational Algebra: \\Helly Numbers and Violator Spaces} 
\author{Jes\'us A. De Loera}
\address{J. De Loera,  Department of Mathematics, 
University of California, Davis.}
\email{deloera@math.ucdavis.edu}

\author{Sonja Petrovi\'c}
\address{S. Petrovi\'c, Applied Mathematics Department, Illinois Institute of Technology.}
\email{Sonja.Petrovic@iit.edu}

\author{Despina Stasi}
\address{D. Stasi, 
 Applied Mathematics Department, Illinois Institute of Technology.}
\email{despina.stasi@gmail.com}
\maketitle

\begin{abstract}
This paper transfers a randomized algorithm, originally used in geometric optimization, to computational problems in commutative algebra.	
We show that Clarkson's sampling algorithm can be applied to two  problems in computational algebra: solving large-scale polynomial systems 
and finding small generating sets of graded ideals. The cornerstone of our work is showing that the theory of violator spaces of G\"artner et al.\ 
applies  to polynomial ideal problems. To show this, one utilizes a Helly-type result for algebraic varieties.
The resulting algorithms have  expected runtime linear in the number of input polynomials,  making  the ideas interesting for 
handling systems with very large numbers of polynomials, but whose rank in the vector space of polynomials is small 
(e.g.,  when the number of variables and degree is constant).
\end{abstract}


\thanks{Keywords: Violator spaces, ideal generators, solving polynomial systems, randomized algorithm in algebra, large sparse systems. \\
AMS subject classification: 68W20, 68R05, 12Y05, 13P10, 14Q10, 08A40.
}

\section{Introduction}

Many computer algebra systems offer excellent algorithms for  manipulation  of polynomials. But despite great success in the field, many algebraic problems have bad worst-case complexity. For example, Buchberger's~\cite{Buchberger65,Buchberger1965Translated,clo} groundbreaking algorithm, key to symbolic computational algebra today, computes a Gr\"obner basis of any ideal, but it has a worst-case runtime that is doubly exponential in the number of variables~\cite{dube}. This presents the following problem: \emph{what should one do about computations whose input is a  very large, overdetermined system of polynomials? In this paper,  we  propose to use randomized sampling algorithms to ease the computational cost in such cases.}

One can argue that much of the success in computation with polynomials (of non-trivial size) often relies heavily on finding specialized structures.
Examples include  Faug{\`e}re's et al.\  fast computation of Gr\"obner bases of  zero-dimensional ideals  \cite{faugere2013polynomial,faugereetal,faugere2014sparse,gao2009counting,lakshman1990complexity},  
specialized software for computing generating sets of toric ideals \cite{4ti2},  several   packages in \cite{M2} built specifically to handle monomial ideals, 
 and the study of sparse systems of polynomials (i.e., systems with  fixed support sets of monomials)  and the associated homotopy methods \cite{sturmfels1991sparse}. A more
recent example of the need to find good structures is in \cite{cifuentes-parrilo}, where Cifuentes and Parrilo began exploiting  chordal graph structure in computational commutative algebra, and in particular, for solving polynomial systems. Our paper exploits combinatorial structure implicit in the input polynomials, but this time  
akin to Helly-type results from convex discrete geometry \cite{Mbook}.

At the same time, significant improvements in efficiency have been obtained by algorithms that involve randomization, rather than deterministic ones (e.g.\ \cite{factoring,primality}); it is also widely recognized that there exist hard problems for which pathological examples requiring exponential runtimes occur only rarely, implying an obvious advantage of considering average  behavior analysis of many algorithms.
For example, some forms of the simplex method for solving  linear programming problems have worst-case complexity that is exponential, yet \cite{spielmanteng2} has recently shown that in the \textit{smoothed analysis of algorithms} sense, the simplex method is a rather robust and fast algorithm.  Smoothed analysis combines the  worst-case and average-case algorithmic analyses by measuring the expected performance of algorithms under slight random perturbations of worst-case inputs. Of course, probabilistic analysis, and smoothed analysis in particular, has been used in computational algebraic geometry for some time now, see  e.g., the elegant work in ~\cite{BePa08,BePa08a,burgissercucker}. 
The aim of this paper is to import a  randomized sampling  framework from geometric optimization  to applied computational algebra, and demonstrate its usefulness on two problems. 

\subsection*{Our contributions}
We apply the theory of violator spaces \cite{ViolatorSpaces2008} to polynomial ideals and  adapt Clarkson's sampling algorithms \cite{c-lvali-95} to provide efficient randomized algorithms 
for the following concrete problems:

{\em
\begin{enumerate}
\item solving  large (overdetermined) systems of multivariate polynomials equations, 
\vskip .1cm

\item  finding small, possibly minimal, generating sets of homogeneous ideals.   
\end{enumerate}
}
\vskip .1cm

Our method is based on using the notion of a \emph{violator space}.  Violator spaces were introduced in 2008 by
G{\"a}rtner, Matou{\v{s}}ek, R{\"u}st, and {\v{S}}kovro{\v{n}} \cite{ViolatorSpaces2008} in a different context. Our approach allows us to adapt  Clarkson's sampling techniques \cite{c-lvali-95} for computation with polynomials. Clarkson-style algorithms rely on  computing with small-size subsystems, embedded in an iterative biased sampling scheme. In the end,  the local information is used to make a global decision about the entire system. The expected runtime  is linear in the number of input elements, which is  the number of polynomials in our case (see \cite{BriseGartnerConf2009} for a more recent simplified version of Clarkson's algorithm for violator spaces). 
Violator spaces naturally appear in problems that have a natural linearization and a sampling size given by a combinatorial \emph{Helly number} of the problem.
While violator spaces and Clarkson's algorithm have already a huge range of applications, to our knowledge, this is the first time such sampling algorithms are being used in computational algebraic geometry. 
For an intuitive reformulation of Helly's theorem for algebraic geometers, see Example~\ref{ex:affineLinear}. Main ingredients of violator spaces are illustrated through  Examples~\ref{ex:intersectRealPlaneCurves}, \ref{ex:intersectRealPlaneCurves-violator} and \ref{ex:intersectRealPlaneCurves-basis+delta+primitive}. A typical setup where problem (1) can be difficult and a randomized algorithm appropriate can be found in Example~\ref{example:coloring}. 

 Before stating the main results, let us fix the notation used throughout the paper. We assume the reader is acquainted with the basics of computational algebraic
geometry as in the award-winning undergraduate textbook \cite{clo}.
Denote by  $K$  a (algebraically closed) field; the reader may keep $K=\mathbb C$ in mind as a running example. 
Let $f_1=0,\dots,f_m=0$ be a system of $m$ polynomials in $n+1$ variables with coefficients in $K$.  We usually assume that $m\gg n$. 
As is customary in the algebra literature, we write $f_1,\dots,f_m\in \RR=K[x_0\dots,x_n]$ and often denote the polynomial ring $\RR$ by a shorthand notation $K[x]$. We will denote by  $\ideal{f_1,\dots,f_m}\subset \RR$ the ideal generated by these polynomials; that is, the set of all polynomial combinations of the $f_i$'s. Note that if $F=\{f_1,\dots,f_m\}$ is a set of polynomials, the ideal $\ideal{f_1,\dots,f_m}$ will equivalently be denoted by $\ideal{F}$. 

A polynomial is said to be homogeneous if all of its terms are of same degree; an ideal generated by such polynomials is a homogeneous ideal. 
   In this paper, the ideals we consider need not be homogeneous; if they are, that will be explicitly stated. 
    In that case, the set of all homogeneous polynomials of total degree $d$ will be denoted by $[\RR]_d$. 
Finally, denote by $\V(S)$ the (affine) variety defined by the set of polynomials $S\in\RR$, that is, the Zariski closure of the set of common zeros of the polynomials in the system $S$.  Therefore, the concrete problem (1) stated above simply asks for the explicit description of the variety (solution set) corresponding to an ideal (system of polynomial equations). The concrete problem (2) asks to find a smaller (e.g., minimal with respect to inclusion) set of polynomial equations that generate the same ideal - and thus have the exact same solution set. 

\subsection*{Solving large polynomial systems} 
Suppose we would like to solve a system of $m$ polynomials in $n+1$ variables over the field $K$, and suppose that $m$ is large.  We are interested in the coefficients of the polynomials as a way to linearize the system.  To that end, recall first that the $d$-th Veronese embedding of $\PP^n$ is the following map $\nu_d:\PP^n \to \PP^{{n+d\choose d}-1}$: 
\begin{align*}
	(x_0:\dots:x_n) &\mapsto (x_0^d:x_0^{d-1}x_1:\dots:x_n^d).
\end{align*}
The map $\nu_d$ induces a coefficient-gathering map for homogeneous polynomials in fixed degree $d$:  
\begin{align*}
	\coeff_d:[\RR]_d &\to K^{n+d\choose d} \\
	\sum_{\alpha:|\alpha|=d}c_\alpha x^\alpha &\mapsto \left[c_{\alpha_1},\dots,c_{\alpha_{n+d\choose d}}\right], 
\end{align*}
where 
$x^{\alpha_i}$ corresponds to the $i$-th coordinate of the $d$-th Veronese embedding. We follow the usual notation $|\alpha|=\sum_i \alpha_i$. 
Therefore, if  $f$ is a homogeneous polynomial of $\deg(f)=d$, $\coeff(f)$ is a vector in the $K$-vector space $K^{n+d\choose d}$. 
This construction can be extended to non-homogeneous polynomials in the following natural way. Consider all distinct total degrees $d_1,\dots,d_s$ of monomials that appear in a non-homogenous polynomial $f$. For each $d_i$, compute the image under $\coeff_{d_i}$ of all monomials of $f$ of degree $d_i$. Finally, concatenate all these vectors into the total coefficient vector of $f$, which we will call $\coeff(f)$ and which is of size ${n+d+1\choose n+1}$, the  number of monomials in $n+1$ variables of (total) degree ranging from $0$ to $d$. 
In this way, a system $f_1,\dots,f_m$ of polynomials in $n$ variables of degree at most $d$ can be represented by its \emph{coefficient matrix} of size ${n+d+1\choose n+1}\times m$. Each column of this matrix corresponds to the vector produced by the map $\coeff_d$ above. This  map allows us to think of polynomials as points in a linear affine space, where Helly's theorem applies. 

We utilize this construction to import  Clarkson's method \cite{c-lvali-95} for solving linear problems to algebraic geometry and, in particular, we make use of Helly-type theorems for varieties. Helly-type theorems allow  one to reduce the problem of solving the system to repeated solution of smaller subsystems, whose size is a Helly number of intersecting linear spaces. As a result, our algorithms achieve {\it expected linear runtime} in the number of input equations. 

%

\begin{thm} \label{thm:solve}
Let $F=\{f_1,\dots,f_m\}\subset\RR$ be a system of  polynomials, 
and let $\delta$ be the dimension of the vector subspace generated by the coefficient vectors of the $f_i$'s, as described above.  

Then there exists a  sampling algorithm that outputs $F'=\{f_{i_1},\dots,f_{i_\delta}\}\subset F$ such that \mbox{$\V(F) = \V(F')$}  in an expected number $O\left(\delta m + \delta^{O(\delta)}\right)$  of calls to the primitive query  that solves a small radical ideal membership problem. 
$F$ and $F'$ generate the  same ideal up to radicals.
\end{thm}
It is important to point out that our sampling algorithm will find a small subsystem  of the input system that, when solved with whatever tools one has at their disposal, will give the same solution set as the original (input) system. Here by `small' we mean a system of size $\delta$, where  $\delta$ is polynomially bounded or constant when the number of variables is constant or when the degree $d$ is small.

%
That the rank $\delta$ of the coefficient matrix of the system of polynomials gives the combinatorial dimension for this problem is shown in Theorem~\ref{thm:F-Helly-vars}. 
There are several interesting special cases of this result. For example, we obtain \cite[Corollary 2]{deza+frankl}  as a corollary: if $f_1,\dots,f_m\in K[x_0,\dots,x_n]$ are homogeneous and of degree at most $d$ each, then the dimension $\delta$ of the vector subspace they generate is at most   ${{n+d}\choose{d}}$ (see Lemma~\ref{thm:DF-Helly}).   Of course, in many situations in practice, this bound is not sharp, as many systems are of low rank. For example, this situation can arise if the monomial support of the system is much smaller than the total number of monomials of degree $d$. In light of this, Theorem~\ref{thm:F-Helly-vars} gives a better bound for low-rank systems. Note that  we measure system complexity by its rank, that is, the vector space dimension $\delta$, and not the usual sparsity considerations such as the structure of the monomial support of the system.  Further, our result applies  to non-homogeneous systems as well.  Its proof is presented in Section \ref{sec:solve},  along with the proof of Theorem~\ref{thm:solve}. 


\subsection*{Computing small generating sets of ideals} 
The problem of finding ``nice'' generating sets of ideals has numerous  applications in statistics, 
optimization, and other fields of science and engineering. 
Current methods of calculating   minimal generating sets of ideals with an a priori large number of generators are inefficient  and  rely 
mostly on Gr\"obner bases computations, since they usually involve ideal membership tests. Of course 
there are exceptional special cases, such as ideals of points in projective space~\cite{Ramella} or binomial 
systems~\cite{4ti2}.   Our second main result  shows how to efficiently extract a small or close to minimal 
generating set for any ideal from a given large generating set and a bound on the size of a minimal generating set. 

\begin{thm} \label{thm:MinGenAlgo} Let $I=\ideal{\mH}$ be an ideal generated  by a (large) finite set of homogeneous  polynomials 
$\mH$, and suppose that $\gamma$ is a known upper bound for the 0-th total Betti number $\beta(R/I)$. 

Then there exists a randomized algorithm that computes a  generating set of $I$ of size $\gamma$ 
in expected number of $\mathcal O(\gamma |\mH|+\gamma^\gamma)$  calls to the  primitive query that solves a small ideal membership problem. 

In particular, if $\gamma=\beta(R/I)$, the algorithm computes a \emph{minimal} generating set of $I$. 
\end{thm}
The proof is presented in Section \ref{sec:mingen}.

\section{A Warm-Up: Algebraic Helly-type theorems and the size of a meaningful sample}

A Helly-type theorem has the following form: 
Given a family of objects $F$, a property $P$, and a \emph{Helly number} $\delta$ such that every subfamily of $F$ with $\delta$ elements has property $P$, then the entire family has property P. 
(See \cite{DGKsurvey63,eckhoff,wenger, ADSHellysurvey}.) 
 In the original theorem of E. Helly, $F$ is a finite family of convex sets in $\R^n$, the constant $\delta$ is  $n+1$, and the property $P$ is  to have a non-empty intersection  \cite{hellyhere}. Here we are looking for non-linear algebraic  versions of the same concept, where the objects in $F$ are algebraic varieties (hypersurfaces) or polynomials; the property desired is to have a common point, or to generate the same ideal;  and the Helly constant $\delta$ will be determined from the structure of the problem at hand. To better understand  the algorithms that we present, it is instructive to consider two intuitive easy examples that highlight  the fundamental combinatorial framework. The first one is an obvious reformulation of Helly's theorem for algebraic geometers. 

\begin{example}\label{ex:affineLinear}
 Let $\mH=\{L_1, L_2, \dots, L_s\}$ be a family of affine linear subspaces in $\R^n$. Consider the case when $s$ is much larger than $n$.   
One would like to answer the following question: when do all of the linear varieties have a nonempty intersection? It is enough to check whether  each subfamily of $\mH$ with $n+1$ elements has a non-empty intersection, as that would imply, by Helly's Theorem, that $\mH$ also has a non-empty intersection. 
Thus, in practice, one can reduce the task of deciding whether $\cap _{i=1}^s L_i \not= \emptyset$ to the collection of smaller queries $\cap_{j=1}^{n+1} L_{i_j}$. However,  instead of testing all possible $ {s \choose n+1}$ many $(n+1)$-tuples,
we may choose to randomly sample multiple $(n+1)$-tuples. Each time we sample, we either verify that one more $(n+1)$-tuple has a non-empty intersection thus increasing the certainty that the property holds for all $(n+1)$-tuples, 
or else find a counterexample, a subfamily without a common point, the existence of which trivially implies that $\cap _{i=1}^s L_i = \emptyset$.
This simple  idea is the foundation of a randomized approach. For now we ask the reader to observe that $n+1$   is the dimension of the vector space of (non-homogeneous) linear polynomials in $n$ variables. 
\end{example}

\begin{example}\label{ex:intersectRealPlaneCurves}
The next example is just slightly more complicated, but illustrates well some key concepts.
Consider next $\mH=\{f_1(x_1,x_2),f_2(x_1,x_2),\dots,f_s(x_1,x_2)\}$, a large family of affine real plane curves of degree at most $d$.
Imagine that $\mH$ is huge, with millions of  constraints $f_i$, but the curves are of small degree, say $d=2$. Nevertheless, suppose that we are in charge of deciding whether the curves in $\mH$ have a common \emph{real} point.
Clearly,  if the pair of polynomials $f,g \in \mH$ intersect, they do so in finitely many points, and, in particular, Bezout's theorem guarantees that no more than $d^2$ intersections occur. One can observe that if the system $\mH$ has a solution, it must pass through some of the (at most $d^2$) points defined by the pair $f,g$ alone. In fact, if we take triples $f,g,h \in \mH$, the same bound of $d^2$ holds, as well as the fact that the solutions for the entire $\mH$ must also be part of the solutions for the triplet $f,g,h$. Same conclusions hold for quadruples, quintuples, and in general $\delta$-tuples. But how large does an integer $\delta$ have to be in order to function as a Helly number? We seek a number $\delta$ such that if all $\delta$-tuples of plane curves in $H$ intersect, then all of the curves in $H$ must intersect. The reader can easily find examples where $\delta=d$ does not work, e.g., for $d \geq 2$. 
\end{example}


To answer the question posed in Example~\ref{ex:intersectRealPlaneCurves}, we refer to Theorem \ref{thm:F-Helly-vars} in Section \ref{sec:solve}. 
Without re-stating the theorem here, we state the following Corollary and note that it gives a nice bound on $\delta$. Corollary~\ref{cor:RealIntersectionHelly} is implied by the observation that there are only ${d+2 \choose 2}$ monomials in two variables of degree $\leq d$ (which
says they span a linear subspace of that dimension inside the vector space of all polynomials) and Theorem \ref{thm:F-Helly-vars}.


%

\begin{corollary} \label{cor:RealIntersectionHelly}
Let $\mH=\{f_1(x,y),f_2(x,y),\dots,f_s(x,y)\}$ be a family of affine real plane curves of degree at most $d$. 
If every $\delta={d+2 \choose 2}$ of the curves have a real intersection point, then all the curves in $\mH$ have a real intersection point. 
If we consider the same problem over the complex numbers, then the same bound holds.
\end{corollary}


Thus, it suffices to check all $\delta$-tuples of curves for a common real point of intersection, 
and if all of those instances do intersect, then we are sure all $|\mH|$ polynomials must have a common intersection point, too. 
The result suggests a brute-force process to verify real feasibility of the system, 
which of course is not a pretty proposition, given that $|\mH|$ is assumed to be very large. Instead, Section~\ref{sec:clarkson} explains 
how to \emph{sample} the set of $\delta$-tuples in order to obtain a solution to the problem more efficiently. 
Notice that it is important to find a small Helly number $\delta$, as a way to find the smallest sampling size necessary to 
detect common intersections. It turns out that in this example and in the case when all $f_i$ are homogeneous, the Helly number is best possible \cite{Frankl}.

\section{Violator spaces and Clarkson's sampling algorithms}\label{sec:clarkson}

The key observation in the previous section was that the existence of Helly-type theorems indicates that there is a natural notion of sampling size to test for a property of varieties. Our goal is to import  to computational algebra an efficient randomized sampling algorithm by Clarkson. To import this algorithm, we use the notion of \emph{violator spaces} which we outline in the remainder of this section.   We illustrate the definitions using Example~\ref{ex:intersectRealPlaneCurves} as a running example in this section. 

In 1992, Sharir and  Welzl \cite{sw-cblpr-92} identified special kinds of geometric optimization  problems that lend themselves to solution via repeated sampling of smaller subproblems: they called these \emph{LP-type} problems. 
Over the years, many other problems were identified as LP-type problems and several abstractions and methods were proposed \cite{a-bbhdn-94,a-httgl-94,bsv-dsapg-01,halman,msw-sblp-92}. A powerful sampling scheme, 
devised by Clarkson \cite{c-lvali-95} for linear programming, works particularly well for geometric optimization problems in small number of variables.  Examples of applications include convex and linear programming, integer linear programming, the problem of computing the minimum-volume ball or ellipsoid enclosing a given point set in  $\R^n$, and the problem of finding the distance of two convex polytopes in $\R^n$.  
In 2008, G\"artner, Matou{\v{s}}ek, R\"ust and \v{S}kovro\v{n} \cite{ViolatorSpaces2008} invented \emph{violator spaces} and showed they give a much more general framework to work with LP-type problems. In fact, violator spaces include all prior abstractions and were proven in \cite{SkovronPhDThesis} to be the most general framework in which Clarkson's sampling converges to a solution. 
Let us begin with the key definition of a violator space.

\begin{defn}[\cite{ViolatorSpaces2008}]
\label{defn:ViolatorSpaces}
A {\it violator space} is a pair $(H,\vi)$, where $H$  is a finite set and $\vi$  a mapping $2^H\to2^H$, such that the following two axioms hold: 
\\\begin{tabular}{ll}
{\it Consistency}: & $G\cap \vi(G)=\emptyset$ holds for all $G\subseteq H$, and\\
{\it Locality}: & $\vi(G)=\vi(F)$ holds for all $F\subseteq G\subseteq H$ such that
$G\cap \vi(F)=\emptyset$.\\
\end{tabular}
\end{defn}

\begin{example}[Example~\ref{ex:intersectRealPlaneCurves}, continued]
\label{ex:intersectRealPlaneCurves-violator}
   To illustrate our definition, we consider Example~\ref{ex:intersectRealPlaneCurves} of $s$ real plane curves $\{f_1,\dots,f_s\}=H$.  

A violator operator for testing the existence of a real point of intersection of a subset $F\subset H$ of the curves should capture the real 
intersection property. One possible way  to  define it is the following map $\vi_{real}:2^H \to 2^H$:  
\[
\vi_{real}(F)=\left\{
h\in H: \V_{\mathbb R}(F) \supsetneq \V_{\mathbb R}(F\cup \{h\})
\right\},
\] 
where $\V_{\mathbb R}(F)$ is the set of common real intersection points of $F$, in other words, the real algebraic variety of $F$. 
Note that, by definition, $\vi_{real}(F)=\emptyset$ if the curves in $F$ have no real points of intersection.
Before explaining why $\vi_{real}$ captures the real intersection property correctly (see Example~\ref{ex:intersectRealPlaneCurves-basis+delta+primitive}),
let us show that the set $(H,\vi_{real})$  is a violator space according to Definition~\ref{defn:ViolatorSpaces}.

Consistency holds by definition of $\vi_{real}$:  for any $h\in F$, $\V_{\mathbb R}(F) = \V_{\mathbb R}(F\cup \{h\})$, and so \mbox{$h\not\in \vi_{real}(F)$.}

To show locality, we begin by showing the auxiliary fact that  $\V_{\mathbb R}(F)=\V_{\mathbb R}(G)$.  The inclusion $\V_{\mathbb R}(G)\subseteq \V_{\mathbb R}(F)$  is direct. This is because for any $S'\subseteq S$ it is always the case that \mbox{$\V_{\mathbb R}(S') \supseteq \V_{\mathbb R}(S)$}.  We need only show $\V_{\mathbb R}(F)\subseteq\V_{\mathbb R}(G)$.  Recall  $F\subseteq G\subseteq H$ and \mbox{$G \cap \vi_{real}(F)=\emptyset$}. Consider $g\in G$.  By the assumption, $g$ does not violate $F$, and the proper containment \mbox{$\V_{\mathbb R}(F) \supsetneq\V_{\mathbb R}(F\cup\{g\})$} does not hold. Therefore the equality $\V_{\mathbb R}(F) = \V_{\mathbb R}(F\cup\{g\})$ must hold. By iteratively adding elements of $G\smallsetminus F$ to $F$ and repeating the argument, we conclude that indeed $\V_{\mathbb R}(F)=\V_{\mathbb R}(G)$.

Finally, we argue that  $\V_{\mathbb R}(F)=\V_{\mathbb R}(G)$ implies \mbox{$\vi_{real}(F)=\vi_{real}(G)$}. We first show $\vi_{real}(F) \subset \vi_{real}(G)$. Take $h \in \vi_{real}(F)$. Then  $\V_{\mathbb R}(G)=\V_{\mathbb R}(F) \supsetneq \V_{\mathbb R}(F\cup \{h\}) =
\V_{\mathbb R}(F) \cap \V_{\mathbb R}(\{h\}) = \V_{\mathbb R}(G) \cap \V_{\mathbb R}(\{h\}) =
 \V_{\mathbb R}(G \cup \{h\})$. The last and third-to-last equalities follow from the fact that for any two sets $S_1, S_2$, one always has $\V_{\mathbb R}(S_1 \cup S_2)=\V_{\mathbb R}(S_1) \cap \V_{\mathbb R}(S_2)$.
The containment $\vi_{real}(F) \supset \vi_{real}(G)$ follows in a similar argument. 
Thus $(H,\vi_{real})$ is a violator space.  	
\end{example}

Every violator space comes equipped with three important components: 
a notion of \emph{basis}, its \emph{combinatorial dimension}, and a \emph{primitive test} procedure. 
We begin with the definition of a basis of a violator space, 
analogous to the definition of a basis of a linear programming problem: a minimal set of constraints that defines a solution space. 

\begin{defn}[{\cite[Definition~7]{ViolatorSpaces2008}}]
Consider a violator space $(H,V)$. $B\subseteq H$  is a  {\it basis} if  $B\cap \vi(F)\neq\emptyset$ holds for all proper subsets $F\subset B$. 
For $G\subseteq H$, a {basis of $G$}  is a minimal subset $B$ of $G$ with $\vi(B)=\vi(G)$. 
\end{defn}

It is very important to note that \emph{a violator operator can capture algebraic problems of interest} as long as the basis for that violator space corresponds to a basis of the algebraic object we study.
Violator space bases come with a natural combinatorial invariant, related to Helly numbers we discussed earlier.
\begin{defn}[{\cite[Definition~19]{ViolatorSpaces2008}}]
The size of a largest basis of a violator space  $(\mH, V)$ is called the {\it combinatorial dimension}  of the violator space and denoted by $\delta=\delta(\mH, V)$.
\end{defn}

A crucial property was proved in \cite{ViolatorSpaces2008}: knowing the violations $\vi(G)$ for all $G \subseteq H$  is enough to compute the largest bases. To do so, one can  utilize Clarkson's randomized algorithm  to compute a basis of a violator space $(H,\vi)$ with $m=|H|$. 
The results about the runtime and the size of the sets involved are summarized below. The primitive operation, used as black box  in all stages of the algorithm, is the 
\emph{violation test primitive}. 
 \begin{defn}\label{defn:primitive} 
 Given a violator space $(\mH,\vi)$,  some set $G\subsetneq \mH$,  and some element $h\in\mH\setminus G$,  the  \emph{primitive} test decides whether $h\in\vi(G)$. 
 \end{defn}
The running example illustrates these three key ingredients. 

\begin{example}[Example~\ref{ex:intersectRealPlaneCurves-violator}, continued]\label{ex:intersectRealPlaneCurves-basis+delta+primitive}
In the example of $s$ real plane curves, the violator operator we defined detects whether the polynomials have a real point of intersection. Note that
a basis would be a (minimal) set of curves $B=\{f_{i_1},\dots,f_{i_\delta}\}$, for some $\delta<s$, such that either  
the curves in $B$ have no real point of intersection, or the real points of intersection of the curves in $B$ are the real intersection of all of $H=\{f_1,\dots,f_s\}$.  
If the set $F$ has no real intersection point, then $\vi_{real}(F)=\emptyset$ by definition, so that set $F$ could be a basis in the sense that it is a certificate of infeasibility for this real-intersection  problem. If, on the other hand, $F$ does have a real intersection point, and $\vi_{real}(F)=\emptyset$, then this means that $F$ is a basis in the sense that the curves in $F$ capture the intersections of all of $H$.  The combinatorial dimension for general $\mH$  is provided by Corollary~\ref{cor:RealIntersectionHelly}, and it equals $\delta={d +2 \choose 2}$. However,  special structure of the curves in $\mH$ may imply a smaller combinatorial dimension. 

The primitive query 
 simply checks, given $f_i\in H$ and a candidate subset $G\subseteq H$, whether the set of real points of intersection of $G\cup\{f_i\}$ 
is smaller than the set of real points of intersection of the curves in $G$ alone.
 The role of the primitive query is therefore not to find a basis directly, but to check, instead, whether a given candidate subset $G$ can be a basis of $H$. This can be done by checking whether $f_i\in\vi_{real}(G)$ for all $f_i\in\mH\setminus G$.   
Clearly, given the primitive test, a basis for $H$ can be found by simply testing all sets of size at most $\delta$, but that would be a waste because the number of times one would need to call the primitive would be $O({|H|}^{\delta+1})$. 

As we will see, this brute-force approach can be avoided.  Namely, in our current example, the randomized algorithm from Theorem~\ref{keytoolvio} below will only sample subsets of \mbox{$\delta={d+2 \choose 2}$} curves from the set $\{f_1,\dots,f_s\}$, and  find a  basis of the violator space of size $\delta$ in the sense explained above. 
\end{example}

The sampling method in \cite{c-lvali-95} avoids a full brute-force approach. It is presented  in two stages, referred to as  Clarkson's first and  second algorithm. We outline these below. 

Clarkson's first algorithm, in the first iteration, draws a small random sample $R \subset G$, calls the second stage to calculate the basis $C$ of $R$, and returns $C$ if it is already a basis for the larger subset $G$. If $C$ is not already a basis, but the elements of $G\setminus C$ violating $R$ are few, it adds those elements to a growing set of violators $W$, and repeats the process with $C$ being calculated as the basis of the set $W\cup R$ for a new randomly chosen small $R\subset G\setminus W$. The crucial point here is that $|R|$ is much smaller than $|G|$ and, consequently, it acts as a Helly number of sorts.  

Clarkson's second algorithm ($\Basis2$) iteratively picks a random small ($6\delta^2$ elements) subset $R$ of $G$, finds a basis $C$ for $R$ by exhaustively testing each possible subset ($\BruteForce$;) taking advantage of the fact that the sample $R$ is very small, 
and then calculates the violators of $G\setminus C$. 
At each iteration, elements that appear in bases with small violator sets get a higher probability of being selected. 

This idea is very important: we are biasing the sampling process, so that some constraints will be more likely to be chosen. This is accomplished by considering every element $h$ of the set $G$ as having a multiplicity ${\mathfrak m}(h)$; the multiplicity of a set is the sum of the multiplicities of its elements. The process is repeated until a basis of $G$ is found, i.e.\ until $\vi(G\setminus C)$ is empty.

\begin{algorithm}
\label{alg:ClarksonsFirst}
\LinesNumbered
\DontPrintSemicolon
\SetAlgoLined
\SetKwInOut{Input}{input}
\SetKwInOut{Output}{output}
\Input{
$G\subseteq \mH$, 
 			$\delta$: combinatorial complexity of $H$ }
\Output{ $\mB$, a basis for $G$ 
}
\BlankLine
\uIf{$|G|\leq9\delta^2$}{
return $\Basis2(G)$\;}
\Else{
$W\leftarrow\emptyset$\;
\Repeat{$V=\emptyset$}{
	$R\leftarrow$ random subset of $G\setminus W$ with $\lfloor\delta\sqrt{|G|}\rfloor$ elements.\;
	$C\leftarrow{\Basis2}(W\cup R)$\;
	$V\leftarrow\{h\in G\setminus C \text{ s.t.}\ h\in \vi(C)\}$\;
	\If{$|V|\leq 2\sqrt{|G|}$}{
		$W\leftarrow W\cup V$}
	}}
return $C$.
\caption{Clarkson's first algorithm}
\end{algorithm}
\smallskip

\begin{algorithm}
\label{alg:ClarksonsFirst}
\LinesNumbered
\DontPrintSemicolon
\SetAlgoLined
\SetKwInOut{Input}{input}
\SetKwInOut{Output}{output}
\Input{
$G\subseteq \mH$; $\delta$: combinatorial complexity of $H$. }
\Output{$\mB$: a basis of $G$}
\BlankLine
\uIf{$|G|\leq 6\delta^2$}{return \BruteForce$(G)$} 
\Else{\Repeat{$V=\emptyset$}{$R\leftarrow$ random subset of $G$ with $6\delta^2$ elements.\;
	$C\leftarrow\BruteForce(R)$\;
	$V\leftarrow\{h\in G\setminus C \text{ s.t.}\ h\in \vi(C)\}$\;
	\If{${\mathfrak m}(V)\leq {\mathfrak m}(G)/3\delta$}{
	\For{$h\in V$}{
		${\mathfrak m}(h)\leftarrow2{\mathfrak m}(h)$\;}}}}
return $C$.
\caption{Clarkson's second algorithm: $\Basis2(G)$}
\end{algorithm}

 Again, as described above, all one needs is to be able  to answer the  {\it Primitive query:} Given $G \subset H$ and $h \in H \setminus G$, decide whether $h \in V(G)$.   The runtime is given in terms of the combinatorial dimension  $\delta(\mH, V)$ and the size of $H$.  The key  result we will use in the rest of the paper concerns the complexity of finding a basis:

\begin{thm}\cite[Theorem~27]{ViolatorSpaces2008} \label{keytoolvio}
Using Clarkson's algorithms, a basis of $\mH$ of a violator space $(\mH,\vi)$ can be found by answering the  primitive query an expected $O\left(\delta \left|\mH\right| + \delta^{O(\delta)}\right)$  times. 
\end{thm}

It is very important to note that, in both stages of Clarkson's method, the query $h\in \vi(C)$ is answered 
via calls to the primitive as a black box. In our algebraic applications, the primitive computation requires solving a small-size subsystem (e.g., via Gr\"obner bases or numerical algebraic geometry methods), or  an ideal membership query applied to the ideal generated by a small subset of the given polynomials. On the other hand, the combinatorial dimension relates to the Helly number of the problem which is usually a number that is problem-dependent  and requires non-trivial mathematical results. 

In the two sections that follow we show how violator spaces naturally arise in non-linear algebra of polynomials. 

\section{A violator space for solving overdetermined systems} \label{sec:solve} 
We discuss our random sampling approach to solve large-size (non-linear) polynomial systems by applying Clarkson's algorithm. In particular, we prove Theorem~\ref{thm:solve} as a corollary of Theorem~\ref{thm:F-Helly-vars}. This result is motivated by, and extends, Helly-type theorems for varieties from \cite{deza+frankl} and \cite{Frankl}, which we use to show   that the above algorithms  apply  to large dense homogeneous systems as well (Corollary~\ref{ClarksonWorksForSolving}). 

First, we define a violator  space that captures (in the sense explained in the previous section)  solvability of a polynomial system. 
\begin{defn}\label{defn:solveViolator}[Violator Space for solvability of polynomial systems]
	Let 	 $S\subset \mH$ be finite subsets of polynomials in $\RR$. Define the violator operator 
	$\vi_{solve}:2^\mH\to2^\mH$ to record the set of polynomials in $H$ which do not vanish on the variety $\V(S)$.  
Formally, 
	\[ \vi_{solve}(S)=\{ f \in\mH : \V(S) \ \text{is not contained in} \ \V(f) \}.
	\]
\end{defn}

\begin{lm} \label{viosolve}
	The pair $(\mH,\vi_{solve})$ is a violator space. 
\end{lm} 
\begin{proof}
Note that $\vi_{solve}(S)\cap S=\emptyset$ by definition of $\vi_{solve}(S)$, and thus the  operator satisfies the consistency axiom. 

To  show locality, suppose that $F\subsetneq G\subset \mH$ and $G\cap \vi_{solve}(F)=\emptyset$.  Since $F\subsetneq G$ we know that  $\V(G) \subseteq \V(F)$. 
On the other hand, by definition, $G \cap \vi_{solve}(F)=\emptyset$ implies that $\V(F) \subseteq \V(g)$ for all $g\in G$. Thus $\V(F)$ is contained in $\bigcap_{g\in G}\V(g)=
\V(G)$.
But then the two varieties are equal. 

To complete the argument we show that 
$\V(F) = \V(G)$ implies $\vi_{solve}(F)=\vi_{solve}(G)$. If $h \in \vi_{solve}(F)$ then $\V(h)$ cannot contain $\V(F)=\V(G)$, thus $h\in \vi_{solve}(G)$ too. The argument is symmetric, hence $\vi_{solve}(F)=\vi_{solve}(G)$. 
\end{proof}
It follows from the definition that the operator $\vi_{solve}$ gives rise to a violator space for which a basis $B$ of $G \subset H$ is a set of polynomials  such that $\V(B)=\V(G)$.  
Therefore, a basis $B\subset G$ will either be a subset of polynomials that has no solution and as such be a \emph{certificate of infeasibility} of the whole system $G$, or it will provide a set of polynomials that are sufficient to find \emph{all common solutions} of $G$, i.e., the variety $\V(G)$. 
\medskip

Next, we need a violation primitive test that decides whether  $h \in \vi_{solve}(F)$, as in Definition~\ref{defn:primitive}. By the definition above, this is equivalent to asking whether $h$ vanishes on all irreducible components of the algebraic variety $\V(F)$. 
As is well known, the points of $\V(F)$ where the polynomial $h$ does not vanish correspond to the variety associated with the saturation ideal $\ideal{\ideal{F}: h^\infty}$.  Thus, we may use ideal saturations for the violation primitive. 
For completeness, we recall the following standard definitions. 
The \emph{saturation} of the ideal $\ideal{F}$ with respect to $f$, denoted by $\ideal{\ideal{F}:f^\infty}$, is defined to be  the ideal of polynomials $g \in\RR$ with $f^mg \in I$ for some
 $m>0$. 
This operation removes from the variety $\V(F)$ the irreducible components on which the polynomial $f$ vanishes. 
Recall that every variety can be decomposed into irreducible components (cf.\ \cite[Section~4.6]{clo} for example). The corresponding algebraic operation is the primary decomposition of the ideal defining this variety.

\begin{lemma}[e.g.\ {\cite[Chapter~4]{atiyahmacdonald}}]
 Let $\cap_{i=1}^m Q_i$ be a minimal primary decomposition
for the ideal $I$. 
The saturation ideal $(I:f^\infty)$
equals $\cap_{f \notin \sqrt{Q_i}
} Q_i$.
\end{lemma}

\begin{proof} 
It is known that $(\cap_{i=1}^m Q_i : f^\infty)=\cap_{i=1}^m(Q_i:                                      
f^\infty)$. We observe further that $(Q_i:f^\infty)=Q_i$ if $f$ does not belong to 
$\sqrt{Q_i}$ and
$(Q_i:f^\infty)=(1)$ otherwise. 
\end{proof}

This allows us to set up the primitive query for $(\mH,\vi_{solve})$. 
 However we do not need to calculate the decomposition explicitly, but can instead carry it out using elimination ideals via Gr\"obner bases, as explained for example in \cite[Exercise 4.4.9]{clo}. 
\begin{obs}\label{SolvePrimitive}
	The primitive query for $(\mH,\vi_{solve})$ 
	 is simply \emph{the saturation test} explained above.
\end{obs} 
 
 \begin{remark}\label{rmk:solve=radical}
 	There is an obvious reformulation of these two  ingredients that is worth stating explicitly. Namely,  since a basis $B$ for the violator space $(\mH,\vi_{solve})$ is a set of polynomials such that $\V(B)=\V(\mH)$, the strong Nullstellensatz implies that  \mbox{$\sqrt{\ideal{B}} = \sqrt{\ideal{\mH}}$.} Thus  a basis determines the ideal of the input system up to radicals, and we could have named the violator operator \mbox{$\vi_{solve}\equiv \vi_{radical}$} instead. 
Furthermore,	a polynomial $h$ vanishing on all irreducible components of the algebraic variety $\V(F)$ is equivalent to $h\in\sqrt{\ideal{F}}$, i.e., $h$ belonging to the radical of the ideal $\ideal{F}$. 
In particular, the primitive query for $\vi_{solve}$ can also be stated as the \emph{radical ideal membership test}. This test can be implemented using Gr\"obner bases, as explained for example in \cite[Proposition 4.2.8]{clo}:  $h\in\sqrt{\ideal{F}}$ if and only if $1\in\ideal{F,1-yh}\subseteq K[x_0,\dots,x_n,y]$.  Therefore,  computation of one Gr\"obner basis  of the ideal $\ideal{F,1-yh}$ suffices to carry out this test. 
\end{remark}

\medskip
Finally, we solve the problem of finding a combinatorial dimension for $\vi_{solve}$. For this, consider, as a warm up, the simpler situation where we have a
Helly-type theorem for hypersurfaces defined by \emph{homogeneous} polynomials. This was proved by Motzkin \cite{motzkin} and then later reproved by 
Deza and Frankl \cite{deza+frankl}, and it provides us with a combinatorial dimension for guaranteeing that a large-scale homogeneous system  has a solution.  Its proof relies on thinking of the polynomial ring $\RR$ as a $K$-vector space (see also the discussion before Definition~\ref{defn:LowRank}). 

\begin{lemma}[\cite{deza+frankl}, Corollary 2]\label{thm:DF-Helly} 
	Let $f_1,\dots,f_m\subset\RR$ be a system of \emph{homogeneous} polynomials, that is, $f_i\in[\RR]_{d_i}$, and  define $d=\max\{d_i\}$. 
	Suppose that every subset of $p={{n+d}\choose{d}}$ polynomials $\{f_{i_1},\dots,f_{i_p}\}\subset \{f_1,\dots,f_m\}$ has a solution. Then the entire system $\{f_1,\dots,f_m\}$ does as well. 	
\end{lemma}

Lemma~\ref{thm:DF-Helly} provides the combinatorial dimension that, along with the variety membership primitive from 
Observation~\ref{SolvePrimitive}, allows us to apply Clarkson's algorithms to the violator  space $(\mH,\vi_{solve})$.
\begin{corollary}\label{ClarksonWorksForSolving}
Let $\ideal{f_1,\dots,f_m}\subset \RR$  be an ideal generated by $m$ homogeneous polynomials in $n+1$ variables of degree at most $d$;  $f_i\in[\RR]_{d_i}$ and $d=\max\{d_i\}$. Let $\delta={{n+d}\choose{d}}$. Then there is an adaptation of Clarkson's sampling algorithm that, in an expected $O\left(\delta m + \delta^{O(\delta)}\right)$ number of calls to the  primitive query~\ref{SolvePrimitive},  computes   
$\{f_{i_1},\dots,f_{i_\delta}\}$ such that $\V(f_1,\dots,f_m) = \V(f_{i_1},\dots,f_{i_\delta})$. 
\end{corollary}
In particular, this algorithm is linear in the number of input equations $m$, and a randomized polynomial time algorithm when the number of variables $n+1$ and the largest degree $d$ are fixed.  
Furthermore, we can extend it to actually \emph{solve} a large system: once a basis $B=\{f_{i_1},\dots,f_{i_\delta}\}$ for the space $(\{f_1,\dots,f_m\},\vi_{solve})$ is found,  then we can use any computer algebra software (e.g.\ \cite{bertini,M2,phcPACK}) to solve $f_{i_1}=\dots=f_{i_\delta}=0$.   
%
%
	
\smallskip
Note that Lemma \ref{thm:DF-Helly} can be thought of as a statement about the complexity of Hilbert's Nullstellensatz.
 If  $\ideal{f_1,\dots,f_m}=\RR$ (i.e.,  $\V(f_1,\dots,f_m)=\emptyset$), then there exists a subset of size $\delta={{n+d}\choose{d}}$ polynomials $\{f_{i_1},\dots,f_{i_{\delta}}\}$ such that $\V(f_{i_1},\dots,f_{i_\delta})=\emptyset$ as well.   In particular, there is a  Nullstellensatz certificate with that many elements. 
The dimension $n+d\choose d$ is, in fact, only an upper bound, attainable only by dense systems.   However, in practice, many systems are very large but sparse, and possibly non-homogeneous.  Let us highlight again that the notion of `sparsity' we consider is captured by a low-rank property of the system of polynomial equations, made explicit below
in terms of the coefficient matrix. This is crucially different from the usual considerations of monomial supports (Newton polytopes) of the system; instead, we look at the coefficients of the input polynomials - that is, we linearize the problem and consider the related vector spaces, as illustrated in the following example.  

\begin{example}\label{example:coloring}
Consider the following system 
consisting of two types of polynomials: polynomials of the form $x_i^2-1$ for $i=1,\dots,n$, and polynomials of the form $x_i+x_j$ for the pairs $\{i,j: i\not\equiv j \mod 2\}$ along with the additional pair $i=1,j=3$.  This system has $m=n^2+n+1$ equations, and the interesting situation is when the number of variables is a large even number, that is, $n=2k$ for any large integer $k$. 
This system of polynomials generates the $2$-coloring ideal of a particular $n$-vertex non-chordal graph. (See \cite{JesDesColoring} and references therein for our motivation to consider this particular system.) 

Consider a concrete graph. Take the $n$-cycle with all possible even chords, and one extra edge $\{1,3\}$. Thus the pairs $\{i,j\}$ are indexed by edges of the  graph $G$ on $n$ nodes where all odd-numbered vertices are connected to all even-numbered vertices, and with one additional edge $\{1,3\}$. 

We wish to decide if the system has a solution, but since there are $n^2+n+1$ many polynomials, 
we would like to try to avoid computing a Gr\"obner basis of this ideal. Instead, we search for a subsystem of some specific size that determines the same variety.  
It turns out that the system actually has no solution. 
Indeed, a certificate for infeasibility is a random subsystem consisting of the first $n$ quadratic equations,  $n-1$ of the edge equations $x_i+x_j$ with $\{i,j: i\not\equiv j \mod 2\}$, and the additional equation $x_1+x_3$.  
For example,  the first $n-1$ edge polynomials  will do to construct a $2n$-sized certificate of this form. 
Why is the number $n+(n-1)+1 = 2n$ so special? 

To answer this question, let us linearize the problem: to each of the polynomials $f$ associate a coefficient (column) vector $\coeff(f)\in\C^{2n+1}$ whose coordinates are indexed by the monomials appearing in the system $x_1^2,\dots,x_n^2,1,x_1,\dots,x_n$.  Putting all these column vectors in one matrix produces the coefficient matrix of the system of the form
\[
	\left[\begin{matrix} 
	I_n&0\\
	{\bf -1} &0 \\
	 0&E
	\end{matrix}\right], 
\]
where $I_n$ is the $n\times n$ identity, ${\bf -1}$ is the row vector with all entries $-1$, and $E$ is the vertex-edge incidence matrix of the graph $G$. 
Since it is known that the rank of an edge-incidence matrix of an $n$-vertex connected graph is $n-1$, the rank of this matrix is $\delta= n+(n-1) +1 = 2n$. 

Remarkably, the magic size of the infeasibility certificate equals the rank of this coefficient matrix. 
\end{example}

This motivating example suggests that the desired Helly-type number of this problem is 
 captured by a natural low-rank property of the  system. To define it precisely,  let us revisit the extension of the Veronese embedding to non-homogeneous polynomials explained in the Introduction. 
 Here we adopt the notation from \cite[Section 2]{Batselier2013} and consider polynomials in $\RR$ of degree up to $d$  as a $K$-vector space denoted by $\mathcal C_d^{n+1}$.   
The vector space $\mathcal C_d^{n+1}$ has dimension ${d+n+1\choose n+1}$, which, of course, equals the number of monomials in $n+1$ variables of (total) degree from $0$ to $d$. 
In this way, any polynomial $f\in\RR$ is represented by a (column) vector, $\coeff(f)\in\mathcal C_d^{n+1}$, whose entries are the \emph{coefficients} of $f$. Thus, any system $S\subset\RR$ defines a matrix with $|S|$ columns, each of which is an element of $\mathcal C_d^{n+1}$.  
\begin{defn}\label{defn:LowRank} \label{defn:sparsity} 
	A system $S\subset \RR$ is said to have \emph{rank  $D$} if $\dim_K\vecsp{S}=D$, where $\vecsp{S}$ is the vector subspace of $\mathcal C_d^{n+1}$ generated by the coefficients of the polynomials in $S$. 
\end{defn}

%
%
%

We need to also make the notion of Helly-type theorems more precise in the setting of varieties. 
\begin{defn}[Adapted from Definition 1.1.\ in \cite{Frankl}] \label{defn:helly}
	A set $S\subset\RR$  is said to have \emph{the $D$-Helly property} if for every nonempty subset $S_0\subset S$, one can find $p_1,\dots,p_D\in S_0$ with $\V(S_0)=\V(p_1,\dots,p_D)$.  
\end{defn}
The following result, which implies Theorem~\ref{thm:solve}, is an extension of \cite{Frankl} to non-homogeneous systems. It also implies (the contrapositive of)  Lemma \ref{thm:DF-Helly} when restricted to  homogeneous systems, in the case when the system has no solution. 
The proof follows that of \cite{Frankl}, although we remove the homogeneity assumption. We include it here for completeness. 
\begin{thm}\label{thm:F-Helly-vars}
	Any polynomial  system $S\subset \RR$ of  rank $D$ has the $D$-Helly property.  
	
	In other words, for all subsets $\P\subset S$, there exist $p_1,\dots,p_{D}\in\P$ such that $\V(\P)=\V(p_1,\dots,p_{D})$. 
\end{thm} 
\begin{proof} 
	Let $\P\subset S$ be an arbitrary subset of polynomials, and denote by $\vecsp{\P}\subset\mathcal C_d^{n+1}$ the vector subspace it generates. 
Let $d_0=\dim_K \vecsp{\P}$. 	We need to find polynomials $p_1,\dots,p_{D}$ such that $\V(p_1,\dots,p_D)=\V(\P)$. Note that $d_0\leq D$, of course, so it is sufficient 
to consider the case $\P=S$. 
	
	Choose a vector space basis $\vecsp{p_1,\dots,p_{D}}=\vecsp{\P}=\vecsp{S}$. 
	It suffices to show $\V(p_1,\dots,p_{D})\subseteq \V(S)$; indeed,  the inclusions $\V(S)\subseteq \V(\P)\subseteq\V(p_1,\dots,p_{D})$ already hold. 
	
	Suppose, on the contrary, there exists 	$x=(x_0,\dots,x_n)\in\C^{n+1}$ and $p\in S$ such that $p(x)\neq 0$ but $p_i(x)=0$ for all $i=1,\dots,D$. Since $p_i$'s generate $S$ as a vector space,  there exist constants $\gamma_i\in K$ with $p=\sum_i\gamma_i p_i$, implying that $p(x)=0$, a contradiction. 
\end{proof}
The proof above is constructive: to find a subset $p_1,\dots,p_D\in S_0$, one only needs to compute a vector space basis for $\vecsp{S_0}$. Thus, linear algebra (i.e., Gaussian elimination) can construct this subset in time $O(|S_0|^3)$. 
The sampling algorithm based on violator spaces is more efficient. 

\begin{proof}[Proof of Theorem~\ref{thm:solve}] From Lemma~\ref{viosolve}, we know that $(\{f_1,\dots,f_m\},\vi_{solve})$ is a violator space. Theorem \ref{thm:F-Helly-vars} shows that it has a combinatorial dimension, and Observation~\ref{SolvePrimitive} shows that there exists a way to answer the primitive test.
Having these ingredients, Theorem~\ref{keytoolvio} holds and it is possible for us to apply Clarkson's Algorithm again. 
\end{proof}

Remark~\ref{rmk:solve=radical} provides the following interpretation of Theorem~\ref{thm:F-Helly-vars}: 
\begin{corollary}\label{cor:radicalsViolator} 
	Let $I=\ideal{f_1,\dots,f_m}\subset\RR$ and let $D=\dim_K\vecsp{f_1,\dots,f_m}$.  Then, for all subsets $\P$ of the generators $f_1,\dots,f_m$, there exist $p_1,\dots,p_{D}\in\P$ such that  $\sqrt{\P}=\sqrt{\ideal{p_1,\dots,p_D}}$.
\end{corollary}


\section{A violator space for finding generating sets of small cardinality} \label{sec:mingen}

In this section, we apply the violator space approach to obtain a version of Clarkson's algorithm for calculating  small generating sets of general homogeneous ideals as defined on page 3.  As in Section~\ref{sec:solve}, this task rests upon three ingredients: the appropriate violator  operator, understanding the combinatorial dimension for this problem, and a suitable primitive query which we will use as a black box. As before, fixing the definition of the violator operator induces the meaning of the word `basis', as well as the construction of the black-box primitive. 

To determine the natural violator space for the ideal generation problem, let $I\subset \RR$ be a \emph{homogeneous} ideal, $\mH$ some initial generating set of $I$, and define the operator $\vi_{SmallGen}$ as follows. 

\begin{defn}\label{defn:idealViolator}[Violator Space for Homogeneous Ideal Generators]
	Let	 $S\subset \mH$ be finite subsets of $\RR$. We define the  operator  $\vi_{SmallGen}:2^\mH\to2^\mH$   to record the set of polynomials in $H$ that are not in the ideal generated by the polynomials in $S$. Formally, 
	\[ \vi_{SmallGen}(S)=\{ f \in\mH : \ideal{S,f}\supsetneq \ideal{S} \}.
	\]
\end{defn}
Equivalently, the operator can be viewed as $\vi_{SmallGen}(S)=\{ f \in\mH : f\not\in \ideal{S} \}$. 

\begin{lm} \label{viosmallgen}
	The pair $(\mH,\vi_{SmallGen})$ is a violator space. 
\end{lm} 
\begin{proof}
Note that $\vi_{SmallGen}(S)\cap S=\emptyset$ by definition of $\vi_{SmallGen}(S)$, and thus the  operator satisfies the  consistency axiom. 

To  show locality, suppose that $F\subsetneq G\subset \mH$ and $G\cap \vi_{SmallGen}(F)=\emptyset$.  Since $F\subsetneq G$,  $\ideal{F}\subseteq \ideal{G}$. On the other hand $G\cap \vi_{SmallGen}(F)=\emptyset$ implies that $G\subseteq\ideal{F}$ which in turn implies that $\ideal{G}\subseteq\ideal{F}.$ Then the ideals are equal. Then, because $\ideal{G}=\ideal{F}$ we
can prove that  $\vi_{SmallGen}(F)=\vi_{SmallGen}(G)$. Note first that $\ideal{G}=\ideal{F}$, 
holds if and only if $\ideal{G,h} = \ideal{F,h}$ for all polynomials $h\in H$. 

Finally, to show  $\vi_{SmallGen}(F)=\vi_{SmallGen}(G)$, we note that $h \in \vi_{SmallGen}(F)$ if and only if $\ideal{G,h}=\ideal{F,h}\supsetneq \ideal{F}=\ideal{G}$, and this chain of equations and containment holds if and only if $h \in \vi_{SmallGen}(G)$. 
Therefore,  locality holds 
as well and 
$\vi_{SmallGen}$ is a violator space operator.
  \end{proof}

It is clear from the definition that $(\mH, \vi_{SmallGen})$ 
is a violator space for which the basis of $G \subset H$ is a minimal generating set of the ideal $\ideal{G}$.  

\medskip
The next ingredient in this problem is the combinatorial dimension: the size of the largest minimal generating set. 
This natural combinatorial dimension already exists in commutative algebra, namely, it equals a certain Betti number. 
(Recall that
  \emph{Betti numbers} are the ranks $\beta_{i,j}$ of modules in the minimal (graded) free resolution of the ring $\RR/I$; see, for example, \cite[Section 1B (pages 5-9)]{eisenbudbook}.) Specifically, the number $\beta_{0,j}$ is defined as the number of elements of degree $j$ required 
among any set of minimal generators of $I$.  The \emph{(0-th) total Betti number} of $\RR/I$, which we will denote by $\beta(\RR/I)$,   simply equals $\sum_j \beta_{0,j}$, and is then the total number of minimal generators of the ideal $I$.  It is well known that while $I$ has many generating sets, every minimal generating set has the same cardinality, namely $\beta(\RR/I)$. In conclusion, it is known that

\begin{obs} 
	The combinatorial dimension for $(H,\vi_{SmallGen})$ is the (0-th) total Betti number of the ideal $I=\ideal{H}$; in symbols,  $\beta(\RR/I)=\delta(\mH,\vi_{SmallGen})$. 
\end{obs}

Although it may be difficult to exactly compute $\beta(\RR/I)$ in general, a natural upper bound for $\beta(\RR/I)$ is the Betti number for any of its initial ideals (the standard inequality holds by upper-semicontinuity; see e.g. \cite[Theorem 8.29]{millersturmfels}).  In particular,  if $\mH$ is known to contain a Gr\"obner basis with respect to some monomial order, 
then the combinatorial dimension can be estimated by computing the minimal generators of an initial ideal of $\ideal{\mH}$, which is a monomial ideal problem and therefore easy.  In general, however, we only need $\beta(\RR/I)<|\mH|$ for the proposed algorithms to  be  efficient. 

\medskip
The last necessary ingredient is the primitive query for $\vi_{SmallGen}$. 
\begin{obs}
	The primitive query for $\vi_{SmallGen}$, deciding  if $h\in \vi_{SmallGen}(G)$ given $h\in\mH$ and $G\subset H$,   is an ideal membership test. 
\end{obs}
  Of course, the answer to the query is usually Gr\"obner-based, but, as before, the size of the subsystems $G$ on which we call the primitive query is small: ($O(\delta^2)$). In fact, it is easy to see that many small Gr\"obner computations for ideal membership cost less than the state-of-the-art, which includes at least one large Gr\"obner computation. 
  
\begin{proof}[Proof of Theorem~\ref{thm:MinGenAlgo}] From Lemma \ref{viosmallgen} we know $(\mH,\vi_{SmallGen})$ is a violator space
and we have shown it has a combinatorial dimension and a way to answer the primitive test.
Having these ingredients, Theorem \ref{keytoolvio} holds and it is possible for us to apply Clarkson's Algorithm.
\end{proof}

\begin{rmk}
Intuitively, the standard algorithm 
for finding minimal generators  needs to at least compute a Gr\"obner basis for an ideal generated by $|H|$ polynomials, and in fact it is much worse than that. One can simplify this by skipping the computation of useless $S$-pairs (e.g.\ as in \cite{gao2003grobner}), but improvement is not by an order of magnitude, overall. The algorithm remains doubly exponential in the size of $H$ for general input. In contrast, our randomized algorithm distributes the computation into many small Gr\"obner basis calculations, where ``many'' means no more than $\mathcal O(\beta |\mH|+\beta^\beta)$, and ``small" means the ideal is generated by only $O(\beta^2)$ polynomials.  
\end{rmk}

To conclude, in a forthcoming paper we will study further the structure of our violator spaces 
and discuss the use of more LP-type methods for the same algorithmic problems. 
We also intend to present some experimental results for the sampling techniques we discussed here. 
We expect  better performance of the randomized methods versus the traditional deterministic algorithms.

\section*{Acknowledgements} We are grateful to Nina Amenta, Robert Ellis, Diane Maclagan,  Pablo Sober\'on,  Cynthia Vinzant, 
and three anonymous referees  for many useful comments and references. 
 The first author is grateful to IMA, the Institute of Mathematics and its Applications of the Univ. of Minnesota for the support received when this
paper was being written. He was also supported by a UC MEXUS grant. This work was supported by the NSF grant  DMS-1522662. 

\bibliographystyle{plain}
\bibliography{randomized-ideals}

\end{document}